\documentclass{article}
\usepackage{amsfonts,amsmath}
\usepackage{amssymb}
\usepackage{amsthm}
\usepackage{paralist}
\usepackage{multirow}
\usepackage{xspace}
\usepackage{graphicx}
\usepackage[numbers]{natbib}
\usepackage[pagebackref=True]{hyperref}
\usepackage[noend,ruled]{algorithm2e}
\usepackage{subcaption}
\usepackage{float}
\usepackage{nicefrac}
\usepackage{mathtools}
\usepackage{tikz}
\usetikzlibrary{patterns,decorations.pathreplacing}
\usetikzlibrary{positioning}
\usepackage{dsfont}
\usepackage{mathtools}
\DeclarePairedDelimiter\ceil{\lceil}{\rceil}
\DeclarePairedDelimiter\floor{\lfloor}{\rfloor}
\usetikzlibrary{shapes}

\makeatletter
\def\NAT@spacechar{~}
\makeatother

\newtheorem{observation}{Observation}

\newcommand{\nash}{Nash equilibrium\xspace}
\newcommand{\nashs}{Nash equilibria\xspace}
\newcommand{\norm}[1]{\|#1\|_1}
\newcommand{\ham}{\Delta}
\newcommand{\dist}{\text{dist}}

\newtheorem{proposition}{Proposition}
\newtheorem{definition}{Definition}
\newtheorem{theorem}{Theorem}
\newtheorem{lemma}{Lemma}
\newtheorem{corollary}{Corollary}

\usepackage{authblk}

\title{Multi-Player Diffusion Games on Graph Classes}
 
\author[1]{Laurent Bulteau\footnote{Laurent Bulteau was supported by the Alexander von Humboldt
    Foundation, Bonn, Germany. Main work done while affiliated with TU~Berlin.}}
\author[2]{Vincent Froese\footnote{Vincent Froese was supported by the DFG, project DAMM (NI 369/13).}}
\author[3]{Nimrod Talmon\footnote{Nimrod Talmon was supported by DFG Research Training Group “Methods for Discrete Structures” (GRK 1408).  Main work done while affiliated with TU Berlin.}}

\affil[1]{\small IGM-LabInfo, CNRS UMR 8049, Universit\'e Paris-Est Marne-la-Vall\'ee, France\\
  \texttt{Laurent.Bulteau@u-pem.fr}}

\affil[2]{\small Institut f\"ur Softwaretechnik und Theoretische Informatik, TU Berlin, Germany\\
  \texttt{vincent.froese@tu-berlin.de}}

\affil[3]{\small Weizmann Institute of Science, Rehovot, Israel\\
  \texttt{nimrodtalmon77@gmail.com}}

\date{} 

\begin{document}

\tikzstyle{p1} = [draw, shape=circle, fill=white, inner sep=2pt]
\tikzstyle{p2} = [draw, shape=circle, fill=gray, inner sep=2pt]
\tikzstyle{p3} = [draw, shape=circle, fill=black, inner sep=2pt]
\tikzstyle{star1} = [draw, shape=star, star points=5, star point ratio=2.5, fill=white, inner sep=1pt]
\tikzstyle{star2} = [draw, shape=star, star points=5, star point ratio=2.5, fill=gray, inner sep=1pt]
\tikzstyle{star3} = [draw, shape=star, star points=5, star point ratio=2.5, fill=black, inner sep=1pt]

\maketitle

\begin{abstract}
We study competitive diffusion games on graphs introduced by
\citet{AFPT10} to model the spread of influence in social networks.
Extending results of \citet{roshanbin14} for two players, we
investigate the existence of pure \nashs for at least three players
on different classes of graphs including paths, cycles, grid
graphs and hypercubes; as a main contribution, we answer an open question proving
that there is no \nash for three players on $m\times n$ grids
with $\min\{m,n\}\geq 5$.
Further,
extending results of \citet{etesami2014complexity} for two players,
we prove the existence of pure \nashs for four players on every~$d$-dimensional hypercube.
\end{abstract}

\section{Introduction}
 Social networks,
 and the diffusion of information within them,
 yields an interesting and well-researched field of study.
 Among other models,
 competitive diffusion games have been introduced
 by~\citet{AFPT10} as a game-theoretic approach towards modelling
 the process of diffusion (or propagation) of influence (or
 information in general) in social networks. Such models have
 applications in ``viral marketing'' where several companies (or brands)
 compete in influencing as many customers (of products) or users (of
 technologies) as possible by initially selecting only a ``small''
 subset of target users that will ``infect'' a large number of other users.
 Herein, the network is modelled as an undirected graph where the
 vertices correspond to the users, with edges modelling influence
 relations between them. The companies,
 being the players of the corresponding diffusion game, choose an
 initial subset of target vertices which then influence other
 neighboring vertices via a certain propagation process.
 More concretely, a vertex adopts a company's product at some specific 
 time during the process if he is influenced by (that is, connected by an
 edge to) another vertex that already adopted this product.
 After adopting a product of one company,
 a vertex will never adopt any other product in the future.
 However, if a vertex gets influenced by several companies at the same
 time, then he will not adopt any of them and he is removed from the
 game. See~\autoref{sec:prelim} for the formal definitions of the game.

 In their initial work, \citet{AFPT10} studied
 how the existence of pure \nashs is influenced by the diameter of the underlying graph.
 Following this line of research,
 \citet{roshanbin14} investigated the existence of \nashs
 for competitive diffusion games with two players on several classes
 of graphs such as paths, cycles and grid graphs. Notably, she proved
 that on sufficiently large grids, there always exists a \nash for
 two players, further conjecturing that there is no \nash for three
 players on grids. We extend the results of
 \citet{roshanbin14} for two players to three or more
 players on paths, cycles and grid graphs, proving
 the conjectured non-existence of a pure \nash for three players on grids as a main
 result. 
 \citet{etesami2014complexity} also followed this line of research,
 by inverstigating the existence of \nashs for competitive diffusion games with two players on $d$-dimensional hypercubes.
 We extend their results by showing that there always exists a \nash for four players on any $d$-dimensional hypercube.
 
 An overview of our results is given in~\autoref{sec:results}.
 After introducing the preliminaries in~\autoref{sec:prelim},
 we discuss our results for paths and cycles in~\autoref{sec:paths_cycles},
 followed by the proof of our main contribution regarding grids in~\autoref{sec:grids}.
 We discuss hypercubes in~\autoref{sec:hypercubes} and
 finish with some statements concerning general graphs in~\autoref{sec:general}.
 
\subsection{Related Work}
\label{sec:relwork}
The study of influence maximization in social networks was initiated by~\citet{kempe2003maximizing}.
Several game-theoretic models have been suggested,
including our model of reference, introduced by~\citet{AFPT10}.
Some interesting generalizations of this model are the model by~\citet{tzoumas2012game},
who considered a more complex underlying diffusion process,
and the model studied by~\citet{etesami2014complexity},
allowing each player to choose multiple vertices.
\citet{durr2007nash} and~\citet{mavronicolas2008voronoi} studied so-called Voronoi games,
which are closely related to our model
(but not similar; there, players can share vertices).
Recently,
\citet{ito2015competitive} considered the competitive diffusion game on weighted graphs,
including negative weights.
Concerning our model,~\citet{AFPT10} claimed the existence of pure
\nashs for any number of players on graphs of diameter at most two,
however,~\citet{THS12} gave a counterexample consisting of a graph
with nine vertices and diameter two with no \nash for two players.

Our main point of reference is the work of~\citet{roshanbin14},
who studied the existence (and non-existence) of pure \nashs mainly
for two players on special graph classes
(paths, cycles, trees, unicycles, and grids);
indeed, our work can be seen as an extension of that work to more than two
players.
\citet{small2012information} already showed that there is a
\nash for any number of players on any star or clique.
\citet{SO13} proved that there is always a pure \nash for two
players on a tree, but not always for more than two players.
\citet{janssen2014finding} considered safe strategies on trees
and spider graphs, where a safe strategy is a strategy which maximizes the minimum payoff
of a certain player, when the minimum is taken over the possible unknown actions of the other players.

\subsection{Our Results}
\label{sec:results}

We begin by characterizing the existence of \nashs for paths and cycles,
showing that, except for three players on paths of length at least six,
a \nash exists for any number of players playing on any such graph
(Theorem~\ref{thm:paths} and~\ref{thm:cycles}).
We then prove Conjecture~1 of~\citet{roshanbin14},
showing that there is no \nash for three players on $G_{m\times n}$,
as long as both $m$ and~$n$ are at least $5$ (\autoref{thm:grids}).
Then,
we prove the existence of \nashs for four players on any $d$-dimensional hypercube (\autoref{thm:hypercubes}).
Finally, we investigate the minimum number of vertices such that
there is an arbitrary graph with no \nash for~$k$ players.
We prove an upper bound showing that there always exists a tree
on~$\floor{\frac{3}{2}k}+2$ vertices with no \nash for~$k$ players (\autoref{thm:general}).

\subsection{Preliminaries}
\label{sec:prelim}

\paragraph*{Notation.}
For~$i, j \in \mathbb{N}$ with~$i < j$, we define $[i, j]:=\{i,\ldots,j\}$
and $[i]:=\{1,\ldots,i\}$.
We consider simple, finite, undirected graphs~$G=(V,E)$
with vertex set~$V$ and edge set~$E\subseteq\{\{u,v\}\mid u,v\in V\}$.
For two vertices~$u,v\in V$, we define the \emph{distance}~$\dist_G(u,v)$ between~$u$ and~$v$ to be the length of a shortest path from~$u$ to~$v$ in~$G$.

A \emph{path}~$P_n=(V,E)$ on~$n$ vertices is the graph
with~$V=[n]$ and $E=\{\{i,i+1\}\mid i\in[n-1]\}$.
A \emph{cycle}~$C_n=(V,E)$ on~$n$ vertices is the graph
with~$V=[n]$ and $E=\{\{i,i+1\}\mid i\in[n-1]\}\cup\{\{n,1\}\}$.
For $m,n \in\mathbb{N}$, the $m\times n$ \emph{grid}~$G_{m\times n}=(V,E)$
is a graph with vertices $V=[m]\times[n]$ and edges
$E=\{\{(x,y),(x',y')\}\mid |x-x'|+|y-y'| = 1\}$.
We use the term \emph{position} for a vertex~$v\in V$.
Note that the distance of two positions~$v=(x,y)$, $v'=(x',y')\in V$
is~$\dist_{G_{m\times n}}(v,v') = \norm{v-v'}:=|x-x'|+|y-y'|$.
For $d\in\mathbb{N}$, $d\ge 1$, the $d$-\emph{dimensional hypercube}~$H_d=(V,E)$ is defined on the vertex set~$V=\{0,1\}^d$, that is, a vertex~$x=x_1\ldots x_d\in V$ is a binary string of length~$d$.
The set of edges is defined as~$E=\{\{x,y\} \mid \ham(x,y)=1\}$, where $\ham(x,y):=|\{i\in[d]\mid x_i \neq y_i\}|$ is the \emph{Hamming distance} of~$x$ and~$y$, that is, the number of positions in which~$x$ and~$y$ have different bits.
Note that~$\dist_{H_d}(x,y)=\ham(x,y)$.

\paragraph*{Diffusion Game on Graphs.}
A \emph{competitive diffusion game}~$\Gamma = (G, k)$ is defined by an undirected
graph~$G=(V,E)$ and a number~$k$ of players (we name the players Player~1,~$\ldots$~,~Player~$k$), each having its
distinct color in~$[k]$.
The \emph{strategy space} of each player is~$V$, such that
each Player~$i$ selects a single vertex~$v_i\in V$ at time~0,
which is then colored by her color~$i$. If two players choose
the same vertex~$v$, then this vertex is removed from the graph.
For Player $i$, we use the terms strategy and position interchangeably,
to mean its chosen vertex.
A \emph{strategy profile} is a tuple~$(v_1,\ldots,v_k)\in V^k$
containing the initially chosen vertex for each player.
The \emph{payoff}~$U_i(v_1,\ldots, v_k)$ of Player~$i$
is the number of vertices with color~$i$ after the following propagation process.
At time~$t+1$, any so far uncolored vertex that has only uncolored neighbors
and neighbors colored in~$i$ (and no neighbors with other colors~$j\in [k]\setminus\{i\}$)
is colored in~$i$. Any uncolored vertex with at least two different
colors among its neighbors is removed from the graph.
The process terminates when the coloring of the vertices
does not change between consecutive steps.
A strategy profile~$(v_1,\ldots,v_k)$ is a (pure) \emph{\nash}
if, for any player~$i\in[k]$ and any vertex~$v'\in V$, it holds that
$U_i(v_1,\ldots,v_{i-1},v',v_{i+1},\ldots,v_k) \le U_i(v_1,\ldots,v_k)$.

\section{Paths and Cycles}
\label{sec:paths_cycles}

In this section,
we fully characterize the existence of \nashs on paths and cycles,
for any number $k$ of players.
\begin{theorem}\label{thm:paths}
  For any $k \in \mathbb{N}$ and any $n \in \mathbb{N}$,
  there is a \nash for~$k$ players on $P_n$, except for $k = 3$ and $n \geq 6$.
\end{theorem}

The general idea of the proof is to pair the players and distribute these pairs evenly.
In the rest of the section,
we prove three Lemmas whose straight-forward combination proves~\autoref{thm:paths}.

  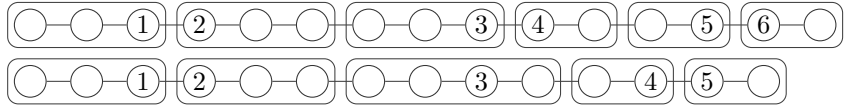
\begin{figure}[t]
    \center
    \begin{tikzpicture}[draw=black!75, scale=0.75,-]
      \tikzstyle{vertex}=[circle,draw=black!80,minimum size=12pt,inner sep=0pt]

      \foreach [count=\i] \pos / \text in {
        {(1, 3)}/,
        {(2, 3)}/,
        {(3, 3)}/$1$,
        {(4, 3)}/$2$,
        {(5, 3)}/,
        {(6, 3)}/,
        {(7, 3)}/,
        {(8, 3)}/,
        {(9, 3)}/$3$,
        {(10,3)}/$4$,
        {(11,3)}/,
        {(12,3)}/,
        {(13,3)}/$5$,
        {(14,3)}/$6$,
        {(15,3)}/
      }
      {
        \node[vertex] (V\i) at \pos {\text};
      }
      
      \draw [rounded corners] (0.6, 2.6) rectangle (3.4, 3.4);
      \draw [rounded corners] (3.6, 2.6) rectangle (6.4, 3.4);
      \draw [rounded corners] (6.6, 2.6) rectangle (9.4, 3.4);
      \draw [rounded corners] (9.6, 2.6) rectangle (11.4, 3.4);
      \draw [rounded corners] (11.6, 2.6) rectangle (13.4, 3.4);
      \draw [rounded corners] (13.6, 2.6) rectangle (15.4, 3.4);

      \foreach \i / \j in {1/2,2/3,3/4,4/5,5/6,6/7,7/8,8/9,9/10,10/11,11/12,12/13,13/14,14/15} {
        \path[] (V\i) edge (V\j);
      }

      \foreach [count=\i] \pos / \text in {
        {(1, 2)}/,
        {(2, 2)}/,
        {(3, 2)}/$1$,
        {(4, 2)}/$2$,
        {(5, 2)}/,
        {(6, 2)}/,
        {(7, 2)}/,
        {(8, 2)}/,
        {(9, 2)}/$3$,
        {(10,2)}/,
        {(11,2)}/,
        {(12,2)}/$4$,
        {(13,2)}/$5$,
        {(14,2)}/
      }
      {
        \node[vertex] (V\i) at \pos {\text};
      }

      \foreach \i / \j in {1/2,2/3,3/4,4/5,5/6,6/7,7/8,8/9,9/10,10/11,11/12,12/13,13/14} {
        \path[] (V\i) edge (V\j);
      }
      
      \draw [rounded corners] (0.6, 1.6) rectangle (3.4, 2.4);
      \draw [rounded corners] (3.6, 1.6) rectangle (6.4, 2.4);
      \draw [rounded corners] (6.6, 1.6) rectangle (10.4, 2.4);
      \draw [rounded corners] (10.6, 1.6) rectangle (12.4, 2.4);
      \draw [rounded corners] (12.6, 1.6) rectangle (14.4, 2.4);
      
    \end{tikzpicture}
    \caption{Illustrations for~\autoref{thm:paths},
    showing
    a \nash for $6$ players on $P_{15}$ (top)
    and a \nash for $5$ players on $P_{14}$ (bottom).
    The boxes show the colored regions of each player.}\label{fig:pathsone}
  \end{figure}

\begin{lemma}\label{lem:evenk}
  For any even $k \in \mathbb{N}$ and any $n\in\mathbb{N}$,
  there is a \nash for~$k$ players on $P_n$.
\end{lemma}

\begin{proof}
  If $n \leq k$,
  then a strategy profile where each vertex of the path is chosen by
  at least one player is clearly a \nash.

  Otherwise, if $n > k$, then the idea is to build pairs of players,
  which are then placed such that two paired players are neighbors
  and the distance of any two consecutive pairs is roughly equal (specifically, differs
  by at most two). See \autoref{fig:pathsone} for an example.
  Intuitively, this yields a \nash since each player
  obtains roughly the same payoff (specifically, differing by at most one),
  therefore no player can improve.
  Since we have $n$ vertices, we want each player's payoff
  to be at least~$z:=\floor{\frac{n}{k}}$.
  This leaves~$r:=n\mod k$ other vertices, which we
  distribute between the first~$r$ players such that the payoff of
  any player is at most~$z+1$.
  This can be achieved as follows.
  Let $p_i \in [n]$ denote the position of Player~$i$,
  that is,
  the index of the chosen vertex on the path.
  We define
  \[
    p_i :=
    \begin{cases}
      z\cdot i +\min\{i,r\} & \text{if } i \text{ is odd}, \\
      p_{i - 1} + 1  & \text{if } i \text{ is even}.
    \end{cases}
  \]
  Note that, by construction, it holds that~$p_1\in \{z,z+1\}$ and
  $p_k=n-z+1$. Moreover, for each odd indexed player~$i\ge 3$, we
  have that~$2z-1\le p_i-p_{i-1}\le 2z+1$.
  We claim that~$u_i:=U_i(p_1,\ldots,p_k)\in\{z,z+1\}$ holds for
  each~$i\in[k]$.
  Clearly, $u_1=p_1\in\{z,z+1\}$ and $u_k=n-p_k+1=z$.
  For all odd~$i\ge 3$, it is not hard to see that
  $u_i=u_{i-1}=1+\floor{(p_i-p_{i-1}-1)/2}\in\{z,z+1\}$, which
  proves the claim.

  To see that the strategy profile~$(p_1,\ldots,p_k)$ is a \nash, consider
  an arbitrary player~$i$ and any other
  strategy~$p_i'\in[n]$ that she picks. Clearly, we can assume $p_i'\neq p_j$
  for all~$j\neq i$ since otherwise Player~$i$'s payoff is zero.
  If~$p_i' < p_1$ or $p_i'>p_k$, then Player~$i$ gets a payoff of at
  most~$z$.
  If~$p_j < p_i' < p_{j+1}$ for some even $j\in[2,k-2]$, then
  her payoff is at most $1+\floor{(p_{j+1}-p_j-2)/2}\le z$.
\end{proof}
We can modify the construction given in the proof of
\autoref{lem:evenk} to also work for odd numbers~$k$ greater than three.

\begin{lemma}
  \label{lem:oddk}
  For any odd $k > 3 \in \mathbb{N}$
  and for any $n\in\mathbb{N}$,
  there is a \nash for $k$ players on $P_n$.
\end{lemma}

\begin{proof}
  We give a strategy profile based on the construction for an even
  number of players (proof of~\autoref{lem:evenk}).
  The idea is to pair the players,
  placing the remaining lonely player between two consecutive pairs.
  
  This is best explained using a reduction to the even case.
  Specifically, given the strategy profile~$(p_1',\ldots,p_{k+1}')$
  for an even number~$k+1$ of players on~$P_{n+1}$
  as constructed in the proof of \autoref{lem:evenk},
  we define the strategy
  profile~$(p_1,\ldots,p_k):=(p_1',\ldots,p_{k-2}',p_k'-1,p_{k+1}'-1)$.
  To see why this  results in a \nash,
  let~$z:=\floor{(n+1)/(k+1)}$ and note that by construction it holds that
  $p_1 \in \{z,z+1\}$,
  $p_k= n-z+1$,
  and~$2z-1\le p_{i+1}-p_i\le 2z+1$ for all~$i\in[2,k-1]$.
  Moreover, each player receives a payoff
  of at least~$z$, therefore all players (except for
  Player~$(k-2)$) cannot improve by the same arguments as in the proof
  of \autoref{lem:evenk}.
  Regarding Player~$(k-2)$, note that her payoff is
  \[1+\floor{(p_{k-1}-p_{k-2}-1)/2} + \floor{(p_{k-2}-p_{k-3}-1)/2}\ge
  2z-1.\]
  Hence, she clearly cannot improve
  by choosing any position outside of~$[p_{k-3},p_{k-1}]$.
  Also, she cannot improve by choosing any other position
  in~$[p_{k-3},p_{k-1}]$. To see this, note that her maximum
  payoff from any position in~$[p_{k-3},p_{k-1}]$ is
  \[1+\floor{(p_{k-1}-p_{k-3}-2)/2}=1+\floor{(p_{k-1}-p_{k-2}-1+p_{k-2}-p_{k-3}-1)/2},\]
  which is equal to the above payoff since~$p_{k-1}-p_{k-2}$ and
  $p_{k-2}-p_{k-3}$ cannot both be even, by construction.
\end{proof}
It remains to discuss the fairly simple (non)-existence of \nashs
for three players.
Note that \citet{roshanbin14} already stated without proof that there
is no \nash for three players on~$G_{2\times n}$ and~$G_{3\times n}$
and \citet{SO13} showed that there is no \nash for three players
on~$P_7$. For the sake of completeness, we prove the following lemma.

\begin{lemma}\label{lem:threeonpath}
  For three players, there is a \nash on~$P_n$ if and only if~$n\le 5$. 
\end{lemma}

\begin{proof}
  If $n \leq 3$, then a strategy profile where each vertex of the path
  is chosen by at least one player is clearly a \nash.
  For $n \in \{4, 5\}$, the strategy profile~$(2,3,4)$ is a \nash.

  To see that there is no \nash for~$n\ge 6$,
  consider an arbitrary strategy profile~$(p_1,p_2,p_3)$.
  Without loss of generality, we can assume that
  $p_1 < p_2 < p_3$ and consider the following two cases.
  First, we assume that $p_2 = p_1 + 1$ and $p_3 = p_2 + 1$.
  If $p_1 > 2$, then Player~2 increases her payoff by choosing $p_1 - 1$.
  Otherwise, it holds that $p_3 < n - 1$ and Player~2 increases her payoff
  by moving to $p_3 + 1$.
  Therefore, this case does not yield a \nash.
  For the remaining case, it holds that $p_1 < p_2 - 1$ or $p_3 > p_2 + 1$.
  If $p_1 < p_2 - 1$, then Player~1 increases her payoff by moving to $p_2 - 1$,
  while if $p_3 > p_2 + 1$, then Player~3 increases her payoff by moving to $p_2 + 1$.
  Thus, this case does not yield a \nash as well, and we are done.
\end{proof}

We close this section with the following result considering cycles.
Interestingly, for cycles there exists a \nash also for three players.

\begin{theorem}
  \label{thm:cycles}
  For any $k, n\in\mathbb{N}$,
  there is a \nash for $k$ players on~$C_n$.
\end{theorem}

\begin{proof}
  It is an easy observation that the constructions
  given in the proofs of Lemma~\ref{lem:evenk} and~\ref{lem:oddk}
  also yield \nashs for cycles, that is,
  when the two endpoints of the path are connected by an edge.
  Thus, it remains to show a \nash for~$k=3$ players for any~$C_n$.
  We set~$p_1:=1$, $p_2:=n$ and
  \[
    p_3 :=
    \begin{cases}
      \floor{n/2} & \text{if } n\mod 4 = 1, \\
      \ceil{n/2}  & \text{else}.
    \end{cases}
  \]
  It is not hard to check that~$(p_1,p_2,p_3)$ is a \nash.
\end{proof}

\section{Grid Graphs}
\label{sec:grids}

In this section we consider three players on the $m\times n$
grid~$G_{m\times n}$
and we prove the following theorem.

\begin{theorem}\label{thm:grids}
  If $m \geq 5$ and $n \geq 5$,
  then there is no \nash for three players on $G_{m\times n}$.
\end{theorem}
Before proving the theorem,
let us first introduce some general definitions
and observations.
Throughout this section, we denote the strategy of Player~$i$,
that is, the initially chosen vertex of Player~$i$,
by~$p_i:=(x_i,y_i)\in[m]\times[n]$.
Note that any strategy profile where more than one player
chooses the same position is never a \nash since
in this case each of these players gets a payoff of zero,
and can improve its payoff by choosing any free vertex
(to obtain a payoff of at least one).
Therefore, we will assume without loss of generality
that $p_1 \neq p_2 \neq p_3$.
Further, note that the game is highly symmetric with respect to the
axes. Specifically, reflecting coordinates along a dimension,
or rotating the grid by~90 degrees, yields the same outcome for the game.
Thus,
in what follows, we only consider possible cases up to
the above symmetries.

We define $\Delta_x:=\max_{i,j\in[k]}|x_i-x_j|$ and
$\Delta_y:=\max_{i,j\in[k]}|y_i-y_j|$ to be the maximum coordinate-wise differences
among the positions of the players.
We say that a player \emph{strictly controls} the other two players,
if both reside on the same side of the player, in both dimensions.  
\begin{definition}
    Player~$i$ \emph{strictly controls} the other players,
    if for each other Player~$j$ with~$j\neq i$, either
    \begin{align*}
      &  (x_i < x_j) \; \land \; (y_i < y_j), \text{ or}\\ 
      &  (x_i < x_j) \; \land \; (y_i > y_j), \text{ or}\\
      &  (x_i > x_j) \; \land \; (y_i < y_j), \text{ or}\\
      &  (x_i > x_j) \; \land \; (y_i > y_j) \text{ holds}. 
    \end{align*}
 \end{definition}
We now prove \autoref{thm:grids}. 

\begin{proof}[Proof of~\autoref{thm:grids}]
  Let~$m \geq 5$ and $n\ge 5$. We perform a case distinction
  based on the relative positions of the three players.
  As a first case, we consider strategy profiles where the players
  are playing ``far'' from each other, that is, there are two players
  whose positions differ by at least four in some coordinate
  (formally $\max\{\Delta_x,\Delta_y\}\ge 3$).
  For these profiles, we distinguish two subcases, namely,
  whether there existss a player who strictly controls the others
  (\autoref{lem:strict}) or not (\autoref{lem:nonstrict}).
  We prove that none of these cases yields a \nash by
  showing that there always exists a player who can improve is payoff.
  Notably, the improving player always moves closer to the other two
  players.
  We are left with the case where the players are playing
  ``close'' to each other, in the sense that their positions all lie
  inside a $3\times 3$ subgrid (that is, $\max\{\Delta_x,\Delta_y\}\le
  2$). For these strategy profiles, we show that there always exists a player who can
  improve her payoff (\autoref{lem:close}), however the improving
  position depends not only on the relative positions between the
  players, but also on the global positioning of this subgrid on the overall grid.
  This leads to a somewhat erratic behaviour, which we overcome
  by considering all possible close positions (up to symmetries) in the
  proof of \autoref{lem:close}.
  Altogether, Lemmas~\ref{lem:strict}, \ref{lem:nonstrict}
  and~\ref{lem:close} cover all possible strategy profiles
  (ruling them out as \nashs),
  thus implying the theorem.\end{proof}

In order to conclude \autoref{thm:grids}, it remains to prove the
lemmas mentioned in the case distinction discussed above.
To this end, we start with two easy preliminary results.
First, we observe that a vertex for which there is a unique player with the shortest distance to it is colored in that player's color (note that this is true in general for every graph and any number of players).

\begin{observation}
  \label{obs:uniqueshortestpath}
  Let~$G=(V,E)$ be an undirected graph and let~$(p_1,\ldots,p_k)$ be a strategy profile.
  Let~$v\in V$ be a vertex for which there exists an~$i\in[k]$ such that~$\dist_G(p_i,v) = :\delta < \dist(p_j,v)$ holds for all~$j\in[k]$, $j\neq i$,
  then~$v$ will be colored in color~$i$ at time~$\delta$.
\end{observation}

\begin{proof}
  The proof is by induction on the distance~$\delta$.
  For~$\delta=0$, it clearly follows from the definition of the propagation process
  that~$v=p_i$ has color~$i$ at time~0.
  For all~$\delta > 0$, it follows from the induction hypothesis that~$v$ has a neighbor~$u$ with~$\dist(p_i,u)=\delta-1$ that is colored in color~$i$ at time~$\delta-1$.
  Moreover, for all neighbors~$w$ of~$v$ it holds~$\dist(p_j,w)>\delta-1$ for all~$j\neq i$.
  This implies that no neighbor of~$v$ has a different color than~$i$ at time~$\delta-1$, and thus, $v$ has color~$i$ at time~$\delta$.
\end{proof}

Based on \autoref{obs:uniqueshortestpath}, we show that whenever a player has
distance at least three to the other players and both
of them are positioned on the same side
of that player
(with respect to both dimensions),
then she can improve her payoff by moving closer to the others
(see \autoref{fig:monotonemove} for an illustration).

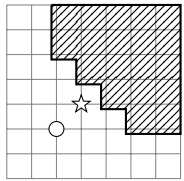
\begin{figure}[t]
    \centering
    \begin{tikzpicture}[scale=.33]
      \draw[help lines] (0,0) grid (7,7);
      \node[p1] at (2, 2) {};
      \draw[thick, pattern = north east lines, pattern color = black] (1.8,7)--(1.8,4.8)--(2.8,4.8)--(2.8,3.8)--(3.8,3.8)--(3.8,2.8)--(4.8,2.8)--(4.8,1.8)--(7,1.8)--(7,7)--(1.8,7);
      \node[star1] at (3,3) {};
    \end{tikzpicture}
    \caption{Example of a strategy profile where Player~1 (white circle) has both other
      players to her top right with distance at least three (the
      shaded region denotes the possible positions for Player~2 and~3). Player~1
      can increase her payoff by moving closer to the others (star).}
    \label{fig:monotonemove}
\end{figure}

\begin{proposition}\label{prop:monotonemove}
  If ~$x_1 \le x_j$, $y_1 \le y_j$, and 
  $\norm{p_1-p_j} \ge 3$ holds for~$j\in\{2,3\}$,
  then Player~1 can increase her payoff by moving to~$(x_1 + 1, y_1 +1)$.
\end{proposition}

\begin{proof}
Let~$p_1':=(x_1+1,y_1+1)$ and~$x\in[x_1]\times[y_1]$.
Note that~$\norm{p_1'-x}=\norm{p_1-x}+2 < \norm{p_j-x} =
\norm{p_1-p_j}+\norm{p_1-x}\ge \norm{p_1-x}+3$ holds for $j\in\{2,3\}$.
Hence, from position~$p_1'$, Player~1 still has the unique shortest distance to~$x$.
By \autoref{obs:uniqueshortestpath},~$x$ gets color~1.
Moreover, for any other position
$x \not\in[x_1]\times[y_1]$, there is a shortest path from~$p_1$
to~$x$ going through at least one of the positions~$(x_1+1,y_1)$,
$(x_1,y_1+1)$, or~$p_1'$.
Clearly, there is also a shortest path from~$p_1'$ to~$x$ of at most the same
length going through one of these positions.
Thus, if~$x$ was colored with color~1 before, then~$x$ is still
colored in color~1.

To see that Player~1 strictly increases her payoff from~$p_1'$, note that
$\norm{p_1'-x}=\norm{p_1-x}-2$ holds for all~$x\in[x_1+1,m]\times [y_1+1,n]$.
Hence, Player~1 now has the unique shortest distance to all those
positions where the distance from~$p_1$ was at most one larger than
the shortest distance from any other player.
We claim that there exists a position~$v\in[x_1+1,m]\times[y_1+1,n]$ with $\norm{p_1-v} - \min_{j\in\{2,3\}}\norm{p_j-v} \in \{0,1\}$
which is not colored in color~1 before.
By \autoref{obs:uniqueshortestpath}, Player~1 then gets~$v$ when moving to~$p_1'$.
To verify the claim, assume first without loss of generality that~$\norm{p_2-p_1}\le\norm{p_3-p_1}$.
Then, on any shortest path from~$p_2$ to~$p_1$, there exists a ``middle'' vertex~$v=(v_x,v_y)\in[x_1,x_2]\times[y_1,y_2]$ with~$\norm{p_1-v}-\norm{p_2-v}\in\{0,1\}$ and~$\norm{p_2-v}\le\norm{p_3-v}$.
Assume towards a contradiction that~$v$ is colored in color~1. Then this implies that on every shortest path from~$p_2$ to~$v$ a vertex has been removed during the propagation process since otherwise either~$v$ would have been removed or colored in color~2.
But this can only happen due to the position of Player~3 since Player~1 has a larger distance than Player~2 to all vertices in the subgrid~$[v_x,x_2]\times[v_y,y_2]$.
Hence, on every shortest path from~$p_2$ to~$v$ there exists a vertex with the same distance to~$p_3$, which implies that~$p_3$ and $p_2$ have the same distance to~$v$.
It follows, that also on every shortest path from~$p_3$ to~$v$ a vertex has to be removed during propagation because of Player~2.
But this is not possible due to the structure of a grid (\autoref{fig:monotonemove_correctness} depicts a typical situation).
Since~$p_2$ and $p_3$ have the same distance to~$v$, it follows from~$p_2\neq p_3$ that~$x_2\neq x_3$ and~$y_2\neq y_3$.
Assume without loss of generality that~$x_2 < x_3$ (and thus~$y_2 > y_3$) and consider the shortest path $v,(v_x,v_y+1),\ldots,(v_x,y_2),(v_x+1,y_2),\ldots,p_2$ from~$v$ to~$p_2$.
At least one of the inner vertices has to be removed during the propagation.
The vertices $(v_x,y_2),\ldots,(x_2-1,y_2)$ cannot be removed because their distance to~$p_3$ is strictly larger than the distance to~$p_2$.
Hence, a vertex $u=(u_x,u_y)\in\{(v_x,v_y+1),\ldots,(v_x,y_2-1)\}$ has to be removed.
Assume that~$u$ has the maximum $y$-coordinate among the removed vertices.
It follows that there is a shortest path from~$p_3$ to~$u$ going through~$(u_x,u_y-1)$ or through~$(u_x+1,u_y)$ on which all inner vertices are colored in color 3.
But this implies that there is also a shortest path from~$p_3$ to~$v$ going through~$(u_x,u_y-1)$ or~$(u_x+1,u_y)$ that is colored with color~3. Hence,~$v$ cannot have color~1.
\end{proof}

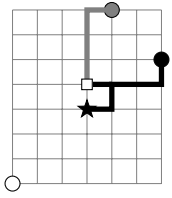
\begin{figure}[t]
    \centering
    \begin{tikzpicture}[scale=.33]
      \draw[help lines] (0,0) grid (6,7);
      \node[p1] at (0, 0) {};
      \node[p2] at (4,7) {};
      \node[p3] at (6,5) {};
      \draw[line width=2pt, gray] (4,7) -- (3,7) -- (3,4);
      \draw[line width=2pt, black] (6,5) -- (6,4) -- (4,4) -- (4,3) -- (3,3);
      \draw[line width=2pt, black] (3,4) -- (4,4);
      \node[draw, rectangle, inner sep=2pt, fill=white] at (3,4) {};
      \node[star3] at (3,3) {};
    \end{tikzpicture}
    \caption{A profile where Player~1 (white) has distance 6 to the vertex~$v$ (black star) and Player~2 (gray) and Player~3 (black) have distance 5 to~$v$. The vertex~$u$ (white square) on the shortest path from Player~2 to~$v$ indicated by thick lines is removed during the propagation due to the shortest path from Player~3 to~$u$. Player~3, however, still reaches~$v$ by a path of length~$5$. Thus, $v$ is not colored by Player~1.}
    \label{fig:monotonemove_correctness}
\end{figure}

We go on to prove the lemmas,
starting with the case that the players play far from each other.
The following lemma handles the first subcase,
that is,
where one of the players strictly controls the others.

\begin{lemma}\label{lem:strict}
    A strategy profile with~$\max\{\Delta_x,\Delta_y\}\ge 3$ where one
    of the players strictly controls the others is not a \nash.
\end{lemma}

\begin{figure}[t]
    \centering
    \begin{subfigure}{.2\linewidth}
      \centering
      \begin{tikzpicture}[scale=.33]
        \draw[help lines] (0,0) grid (5,5);
        \node[p1] at (0, 0) {};
        \draw[thick, pattern = north east lines, pattern color = black] (0.8,5)--(0.8,2.8)--(1.8,2.8)--(1.8,0.8)--(5,0.8)--(5,5)--(0.8,5);
        \node[star1] at (1,1) {};
      \end{tikzpicture}
      \caption*{1}
      \label{fig:case1yes}
    \end{subfigure}
    \begin{subfigure}{.2\linewidth}
      \centering
      \begin{tikzpicture}[scale=.33]
        \draw[help lines] (0, 0) grid (5,5);
        \node[p1] at (0, 0) {};
        \node[p2] at (1, 1) {};
        \draw[fill=black] (1.8,5)--(1.8,2.8)--(2.8,2.8)--(2.8,1.8)--(5,1.8)--(5,5)--(1.8,5);
        \node[star3] at (2,2) {};
      \end{tikzpicture}
      \caption*{2(a)}
      \label{fig:case2ayes}
    \end{subfigure}
    \begin{subfigure}{.2\linewidth}
      \centering
      \begin{tikzpicture}[scale=.33]
        \draw[help lines] (0,0) grid (5,5);
        \node[p1] at (0, 0) {};
        \node[p2] at (1, 1) {};
        \draw[line width = 4pt] (1,2.8)--(1,5);
        \node[star3] at (1,2) {};
      \end{tikzpicture}
      \caption*{2(b)}
      \label{fig:case2byes}
    \end{subfigure}
    \caption{Possible cases (up to symmetry) for Player~1 (white) strictly
      controlling Player~2 (gray) and Player~3 (black).
      Circles denote the player's strategies. The shaded region
      contains the possible positions of both Player~2 and~3, whereas the
      black regions denote possible positions for Player~3 only.
      A star marks the position improving the payoff of the respective player.}
    \label{fig:strict}
  \end{figure}
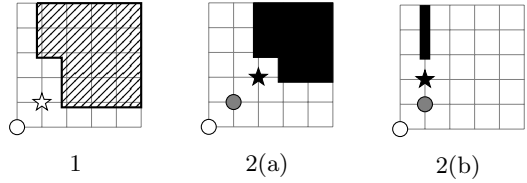

\begin{proof} 
  We assume without loss of generality that
  Player~1 strictly controls Player~2 and Player~3,
  specifically,
  we assume that $x_1 < x_2$ and $y_1 < y_2$ and $x_1 < x_3$ and $y_1 < y_3$ holds.
  \autoref{fig:strict} depicts the three possible cases
  for the positions of Player~2 and Player~3.
  For each case, we show that a player which can improve her payoff exists.
  \begin{enumerate}[\bf {Case} 1{:}]
    \item
      We assume that $(x_2, y_2) \neq (x_1 + 1, y_1 + 1)$ and $(x_3, y_3) \neq (x_1 + 1, y_1 + 1)$.
      By~\autoref{prop:monotonemove}, Player~1 gets a higher payoff
      from $(x_1 + 1, y_1 + 1)$.
    \item
      We assume without loss of generality that $(x_2, y_2) = (x_1 +
      1,y_1 + 1)$.
      \begin{enumerate}[\bf {(}a{)}]
        \item
          We assume~$x_2 < x_3$ and~$y_2 < y_3$.
          Then, $x_3>x_2+1$ or~$y_3>y_2+1$ holds
          since $\max\{\Delta_x,\Delta_y\}\ge 3$.
          Note that Player~3 strictly controls Player~1 and Player~2
          and that this case is symmetric to Case~1.
        \item\label{case:2b}
          We assume $x_2 \geq x_3$ or $y_2 \geq y_3$.
          Then, it holds that $x_3=x_2$ or~$y_3=y_2$.
          We assume~$x_3=x_2$ (the argument for $y_3=y_2$ being analogous).
          Since $\max\{\Delta_x,\Delta_y\}\ge 3$, we have $y_3 > y_2+1$,
          thus Player~3 can improve by moving to $(x_2, y_2 + 1)$
          because then all positions in $[m]\times[y_2+1,n]$
          are colored in color~3, and before only a strict
          subset of these positions were colored in her color.
      \end{enumerate}
  \end{enumerate}
\end{proof}
The other subcase,
where no player strictly controls the others,
is handled by the following lemma.

\begin{lemma}\label{lem:nonstrict}
    A strategy profile with $\max\{\Delta_x,\Delta_y\}\ge 3$ where
    no player strictly controls the others is not a \nash.
  \end{lemma}
  
\begin{proof}
  If no player strictly controls the others, then it follows
  that at least two players have the same coordinate in at least one
  dimension. We perform a case distinction on the cases as depicted in
  \autoref{fig:nonstrict}.

  \begin{enumerate}[\bf {Case} 1{:}]
    \item
      All three players have the same coordinate in one dimension.
      We assume that~$x_1 = x_2 =x_3$ (the case~$y_1=y_2=y_3$ is
      analogous). Without loss of generality also~$y_1 < y_2 < y_3$ holds.
      Since $\max\{\Delta_x,\Delta_y\}\ge 3$, it follows
      that~$y_{i+1}-y_i\ge 2$ holds for some~$i\in\{1,2\}$, say
      for~$i=2$.
      Clearly, Player~3 can improve her payoff by
      choosing~$(x_3,y_2+1)$ (analogous to Case~\ref{case:2b}
      in the proof of \autoref{lem:strict}).
    \item
      There is a dimension where two players have the same
      coordinate but not all three players have the same coordinate in
      any dimension. We assume~$x_1=x_2 < x_3$ and~$y_1 < y_2$ (all other cases are
      analogous). We also assume that~$y_1 \le y_3 \le y_2$,
      since otherwise Player~3 strictly controls the others,
      and this case is handled by~\autoref{lem:strict}.
      \begin{enumerate}[\bf {(}a{)}]
        \item
          We assume that $y_2 = y_1 + 1$.
          Then~$x_3 \ge x_1 + 3$ holds since $\max\{\Delta_x,\Delta_y\}\ge 3$.
          Player~3 increases her payoff by moving to $(x_1 + 2, y_1)$
          (analogous to Case~\ref{case:2b} in the proof of \autoref{lem:strict}).
        \item
          We assume that $y_2 = y_1 + 2$.
          Then~$x_3 \ge x_1 + 3$ holds since $\max\{\Delta_x,\Delta_y\}\ge 3$.
          Player~3 increases her payoff by moving to $(x_1 + 2, y_1+1)$
          (analogous to Case~\ref{case:2b} in the proof of \autoref{lem:strict}).
        \item
          We assume that $y_2 > y_1 + 2$ and $|y_2 - y_3| \leq |y_1-y_3|$.
          That is, without loss of generality,
          Player~3 is closer to Player~2.
          Then, by~\autoref{prop:monotonemove}, Player~1 increases her
          payoff by moving to $(x_1 + 1, y_1 + 1)$.
      \end{enumerate}
  \end{enumerate}
\end{proof}

\begin{figure}[t]
  \centering
  \begin{subfigure}{.2\linewidth}
    \centering
    \begin{tikzpicture}[scale=.33]
      \draw[help lines] (0,0) grid (5,5);
      \node[p1] at (0, 0) {};
      \node[p2] at (0, 1) {};
      \draw[fill=black] (0,2.8) rectangle (0.3,5);
      \node[star3] at (0, 2) {};
    \end{tikzpicture}
    \caption*{1}\label{fig:case1no}
  \end{subfigure}
  \begin{subfigure}{.2\linewidth}
    \centering
    \begin{tikzpicture}[scale=.33]
      \draw[help lines] (0,0) grid (5,5);
      \node[p1] at (0, 0) {};
      \node[p2] at (0, 1) {};
      \draw[fill=black] (2.8,0) rectangle (5,1.2);
      \node[star3] at (2, 0) {};
    \end{tikzpicture}
    \caption*{2(a)}\label{fig:case2ano}
  \end{subfigure}
  \begin{subfigure}{.2\linewidth}
    \centering
    \begin{tikzpicture}[scale=.33]
      \draw[help lines] (0,0) grid (5,5);
      \node[p1] at (0, 0) {};
      \node[p2] at (0, 2) {};
      \draw[fill=black] (2.8,0) rectangle (5,2.2);
      \node[star3] at (2, 1) {};
    \end{tikzpicture}
    \caption*{2(b)}\label{fig:case2bno}
  \end{subfigure}
  \begin{subfigure}{.2\linewidth}
    \centering
    \begin{tikzpicture}[scale=.33]
      \draw[help lines] (0,0) grid (5,5);
      \node[p1] at (0, 0) {};
      \node[p2] at (0, 3) {};
      \draw[fill=black] (0.8,1.8) rectangle (5,3.2);
      \node[star1] at (1,1) {};
    \end{tikzpicture}
    \caption*{2(c)}\label{fig:case2cno}
  \end{subfigure} 
  \caption{Possible cases (up to symmetry) when no player strictly
    controls the others. Circles denote the positions of Player~1
    (white) and Player~2 (gray).
    The black regions contain the possible positions for Player~3.
    A star marks the position improving the payoff of the respective
    player.}
  \label{fig:nonstrict}
\end{figure}
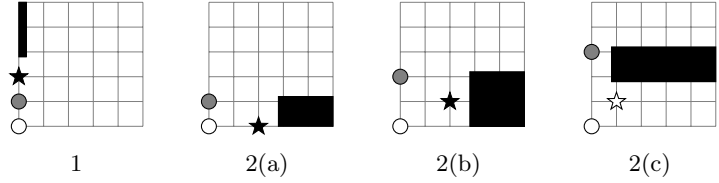
It remains to consider the cases where the players play close to each other.

\begin{lemma}\label{lem:close}
  A strategy profile with $\max\{\Delta_x, \Delta_y\}\le 2$ is not a \nash.
\end{lemma}

\begin{proof}
  First,
  we assume that~$\Delta_x + \Delta_y\ge 2$,
  as otherwise there would be at least two players on the same position
  (so each one of them can improve by moving to any free vertex).
  Without loss of generality,
  we also assume that~$\Delta_x \le \Delta_y$,
  leaving the following cases to consider (depicted in~\autoref{fig:caseslemclose}).
    
\begin{enumerate}[\bf {Case} 1{:}]

\item
Let~$(x_1,y_1) = (x,y) \in [m]\times[n-2]$,
$(x_2,y_2)=(x,y+1)$, and $(x_3,y_3)=(x,y+2)$
be the positions of the three players.
Clearly, with these positions, Player~1
gets a payoff of~$my$, Player~2's payoff is~$m$, and
Player~3 gets a payoff of~$m(n-y-1)$.
Now, if~$y\ge 3$, then Player~2 can improve by
choosing~$(x,y-1)$ thus achieving a payoff of~$m(y-1)>m$.
If~$y=1$, then Player~2 improves by playing on~$(x,y+3)$
with a payoff of~$m(n-y-2) > m$ (remember that~$n\ge 5$).
Also, if~$y=2$ and~$n>5$, then Player~2 gets a higher payoff
by choosing~$(x,y+3)$.
For~$y=2$ and~$n=5$, we observe that either~$x\ge \lceil m/2\rceil$
or~$m-x\ge \lceil m/2\rceil$ and assume,
without loss of generality,
that~$x\ge \lceil m/2\rceil$ holds.
Then, applying~\autoref{obs:uniqueshortestpath},
it holds that, if Player~2 chooses~$(x-1,y)$,
then all positions in~$[x-1]\times [y+1]$ are colored in color~2,
Thus, the payoff of Player~$2$ is at least~$(\lceil m/2\rceil-1)3 >
m$ (remember that~$m\ge 5$).

  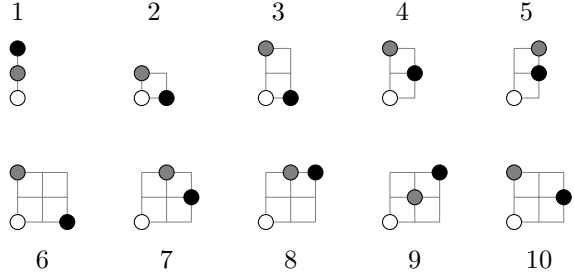
\begin{figure}[t]
    \centering 
    \begin{tikzpicture}[scale=.33]
      \draw[help lines] (1,6) grid (1,8);
      \node[p1] at (1,6) {};
      \node[p2] at (1,7) {};
      \node[p3] at (1,8) {};
      \node at (1,9.5) {1};
      \draw[help lines] (6,6) grid (7,7);
      \node[p1] at (6,6) {};
      \node[p2] at (6,7) {};
      \node[p3] at (7,6) {};
      \node at (6.5,9.5) {2};
      \draw[help lines] (11,6) grid (12,8);
      \node[p1] at (11,6) {};
      \node[p2] at (11,8) {};
      \node[p3] at (12,6) {};
      \node at (11.5,9.5) {3};
      \draw[help lines] (16,6) grid (17,8);
      \node[p1] at (16,6) {};
      \node[p2] at (16,8) {};
      \node[p3] at (17,7) {};
      \node at (16.5,9.5) {4};
      \draw[help lines] (21,6) grid (22,8);
      \node[p1] at (21,6) {};
      \node[p2] at (22,8) {};
      \node[p3] at (22,7) {};
      \node at (21.5,9.5) {5};
      \draw[help lines] (1,1) grid (3,3);
      \node[p1] at (1,1) {};
      \node[p2] at (1,3) {};
      \node[p3] at (3,1) {};
      \node at (2,-0.5) {6};
      \draw[help lines] (6,1) grid (8,3);
      \node[p1] at (6,1) {};
      \node[p2] at (7,3) {};
      \node[p3] at (8,2) {};
      \node at (7,-0.5) {7};
      \draw[help lines] (11,1) grid (13,3);
      \node[p1] at (11,1) {};
      \node[p2] at (12,3) {};
      \node[p3] at (13,3) {};
      \node at (12,-0.5) {8};
      \draw[help lines] (16,1) grid (18,3);
      \node[p1] at (16,1) {};
      \node[p2] at (17,2) {};
      \node[p3] at (18,3) {};
      \node at (17,-0.5) {9};
      \draw[help lines] (21,1) grid (23,3);
      \node[p1] at (21,1) {};
      \node[p2] at (21,3) {};
      \node[p3] at (23,2) {};
      \node at (22,-0.5) {10};
    \end{tikzpicture}
    \caption{Possible positions (up to symmetry) of three players playing inside a
      subgrid of size at most~$3\times 3$. The position of Player~1
      (white) is denoted~$(x,y)$.}
    \label{fig:caseslemclose}
  \end{figure}
  
\item
Let~$(x_1,y_1)=(x,y)\in[m-1]\times[n-1]$,
$(x_2,y_2)=(x,y+1)$, and $(x_3,y_3)=(x+1,y)$
be the positions of the three players.
Note that,
due to symmetry,
this is the only case we have to consider (without loss of generality).
Clearly, exactly the positions in~$[x]\times[y]$ are colored in color~1,
therefore Player 1 has a payoff of~$xy$.
If~$x< m/2$, then position~$(x+2,y)$ yields a
payoff of $(m-x-1)n\ge (\lceil m/2\rceil-1)n$ for Player~1 since she colors all vertices in
$[x+2,m]\times [n]$. Note that~$\lceil m/2\rceil-1 \ge x$ and~$n>y$, thus
Player~1's payoff improved.
Analogously, for~$y< n/2$, Player~1's payoff
from position~$(x,y+2)$ is at least~$m(\lceil n/2\rceil-1) >xy$.
If~$x > \lceil m/2\rceil$, then Player~3 can improve by choosing~$(x-1,y)$.
To see that this is true, note first that by
choosing~$(x_3,y_3)$, Player~3 colors only positions in~$[x+1,m]\times[n]$.
Now, observe that if position~$(x',y')$ is colored in color~3
when Player~3 chooses~$(x_3,y_3)$, then
it holds also that position~$(x-(x'-x),y)$ is colored in color~3
when Player~3 chooses~$(x-1,y)=(x-(x_3-x),y)$,
due the symmetries,
and since the distances from this position to the players positions
are all identical.
Hence, Player~3 colors at least the same number of
positions in~$[x-(m-x),x-1]\times[y]$.
Note that~$x-(m-x) \ge 2$ since~$x > \lceil m/2 \rceil$.
By \autoref{obs:uniqueshortestpath},
Player~3 additionally colors position~$(1,1)$
yielding a strictly greater payoff.
By analogous arguments, for~$y > \lceil n/2\rceil$, the position~$(x,y-1)$ yields
a better payoff for Player~2.
Finally, assume~$x=\lceil m/2\rceil$ and~$y=\lceil
n/2\rceil$.
Now, Player~1 can improve by choosing~$(x-1,y+1)$,
thus coloring at least all positions in~$[x-1]\times[n]$,
giving a payoff of at least~$(\lceil m/2\rceil-1)n > \lceil
m/2\rceil\cdot \lceil n/2\rceil=xy$ for~$m\ge 5$ and~$n\ge 5$.

\item
Let~$(x_1,y_1)=(x,y)\in[m-1]\times[n-2]$, $(x_2,y_2)=(x,y+2)$, and
$(x_3,y_3)=(x+1,y)$ be the positions of the three players.
First, note that Player~1 colors all positions
in~$[x]\times[y]$, gaining a payoff of~$xy$, Player~2 colors
$[m]\times[y+2,n]$, gaining a payoff of~$m(n-y-1)$, and
Player~3 gets all positions in~$[x+1,n]\times[y+1]$, gaining a
payoff of exactly~$(m-x)(y+1)$.
Now, if~$y=1$, then Player~1 can choose~$(x,4)$ to obtain a
payoff of~$m(n-3) > x$, since~$m>x$ and~$n-3\ge 2$.
Hence, assume~$y > 1$, and observe that both the payoff of Player~1
from~$(x+1,y-1)$ and the payoff of Player~3 from~$(x,y-1)$ equals
$m(y-1)$. Assuming that we have a \nash, we obtain the two
inequations~$xy \ge m(y-1)$ and~$(m-x)(y+1)\ge m(y-1)$, which
yield~$m(y-1)/y\le x \le 2m/(y+1)$.
Note that we obtain a contradiction for~$y\ge 3$.
Hence, we can assume that~$y=2$ and~$\lceil m/2\rceil \le x \le \lfloor
2m/3\rfloor$.
If~$n\ge 6$, then Player~1 can improve by choosing~$(x,5)$ achieving
a payoff of~$m(n-4)\ge 2m > 2x$. Thus, we also assume~$n=5$.
Now, Player~1 can choose position~$(x-1,4)$ to color all but
three positions in~$[x-1]\times[5]$. The only positions which she does
not color are~$(x-1,2)$, $(x-1,1)$ and~$(x-2,1)$.
Her payoff is thus~$5(x-1)-3$,
which, for all~$x\ge\lceil m/2\rceil \ge 3$,
is more than~$2x$.
    
\item
Let~$(x_1,y_1)=(x,y)\in[m-1]\times[n-2]$, $(x_2,y_2)=(x,y+2)$, and
$(x_3,y_3)=(x+1,y+1)$ be the positions of the three players.
It is easy to see that, if $x = m-1$,
then Player~3's payoff is exactly one,
and she can gain more by choosing~$(x-1,y+1)$ instead.
For $x < m-1$, note that, apart from~$(x_3,y_3)$, Player~3
colors only positions~$(x',y')$ with~$x'\ge x_3+1=x+2$.
Note also that Player~3 does not color all of these positions.
For example, at least one of the positions~$(x+2,y-1)$
or~$(x+2,y+3)$ exists on the grid (since $n\ge 5$) and is
reached by Player~1 or Player~2 at the same time during the
propagation process of the game.
However, by choosing~$(x+2,y+1)$, Player~3 still colors
the position~$(x_3,y_3)$ and clearly all positions~$(x',y')$
with~$x'\ge x+2$, thus improving her payoff.

\item
Let~$(x_1, y_1) = (x, y) \in [m- 1] \times [n- 2]$,
$(x_2, y_2) = (x + 1, y + 2)$,
and $(x_3, y_3) = (x + 1, y + 1)$
be the positions of the three players.
If $y = n - 2$, then Player~2's payoff is exactly~$m$.
Therefore, she increases her payoff by moving to $(y - 1, x)$,
because then her payoff is at least~$2m$.
Thus, we can assume that $y < n - 2$.
If $x \geq m - 2$,
then Player~3's payoff is either~1 or~3
(it is 1 if $x = m - 1$ and 3 if $x = m - 2$).
Therefore, she increases her payoff by moving to $(x - 2, y)$,
because then her payoff is at least~$n$ (and $n \geq 5$).
Thus, we can assume that $x < m - 2$.
We are left only with the case where $y < n - 2$ and $x < m - 2$.
In this case, Player~3 increases her payoff by moving to $(x +
2, y + 2)$.

\item
Let~$(x_1, y_1) = (x, y)\in[m-2]\times[n-2]$,
$(x_2, y_2) = (x, y + 2)$,
and $(x_3, y_3) = (x + 2, y)$
be the positions of the three players.
It is clear that only a strict subset of the positions in~$[x+2,m]\times[n]$
are colored in color~3 (for example, the position~$(x+2,y+2)$ is not colored in color~3).
By choosing~$(x+2,y+1)$, however, all positions in~$[x+2,m]\times[n]$
are colored in color~3, resulting in a strictly higher payoff for Player~3.
    
\item
Let~$(x_1, y_1) = (x, y)\in[m-2]\times[n-2]$,
$(x_2, y_2) = (x + 1, y + 2)$,
and $(x_3, y_3) = (x + 2, y + 1)$
be the positions of the three players.
By~\autoref{prop:monotonemove},
Player~1 increases her payoff by moving to~$(x_1+1, y_1+1)$.

\item
Let~$(x_1, y_1) = (x, y)\in[m-2]\times[n-2]$,
$(x_2, y_2) = (x + 1, y + 2)$,
and $(x_3, y_3) = (x + 2, y + 2)$
be the positions of the three players.
By~\autoref{prop:monotonemove},
Player~1 increases her payoff by moving to~$(x_1+1, y_1+1)$.
         
\item
Let~$(x_1, y_1) = (x, y)\in[m-2]\times[n-2]$,
$(x_2, y_2) = (x + 1, y + 1)$,
and $(x_3, y_3) = (x + 2, y + 2)$
be the positions of the three players.
Notice that the payoff of Player~2 is only one.
It is clear that Player~2 increases her payoff more by moving to~$(x, y + 1$),
because then her payoff is at least two,
as she also colors the position~$(x, y + 2)$.

\item
Let~$(x_1, y_1) = (x, y)\in[m-2]\times[n-2]$,
$(x_2, y_2) = (x, y + 2)$,
and $(x_3, y_3) = (x + 2, y + 1)$
be the positions of the three players.
Note first that Player~1's colors exactly the positions in~$[x+1]\times[y]$,
thus, her payoff is~$(x+1)y$.
For $y = 1$, Player~1 can move to~$(x,4)$, achieving a payoff
of at least~$(x+1)(n-3)\ge 2(x+1)$ since she colors all
positions in~$[x+1]\times[4,n]$.
Otherwise, if $y \ge 2$,
then, by choosing~$(x+1,y)$, Player~1 still colors all
positions in~$[x+1]\times[y]$, and additionally also the position~$(x+2,y-1)$.
\end{enumerate}
\end{proof}

\section{Hypercubes}
\label{sec:hypercubes}

\citet{etesami2014complexity} studied diffusion games on hypercubes and
proved the existence of a \nash for two players on every~$d$-dimensional hypercube.
In this section, we extend their result to four players.

Recall that the vertices of~$H_d$ are all binary strings of length~$d$, where two vertices are adjacent if and only if they differ in exactly one bit. Moreover, the geodesic distance of two vertices~$u$ and~$v$ is exactly the Hamming distance~$\ham(u,v)$. 
By~$\overline{a}$, we denote the \emph{complement} of~$a=a_1\ldots a_d$, where~$\overline{a}_i := 1-a_i$ for all~$i\in[d]$.

We prove the following theorem.

\begin{theorem}\label{thm:hypercubes}
  Let~$x,y\in\{0,1\}^d$ be two adjacent vertices of~$H_d$ for $d\ge 1$.
  Then, every strategy profile~$(p_1,\ldots,p_4)$ with $\{p_1,\ldots,p_4\}=\{x,\overline{x},y,\overline{y}\}$ is a \nash.
\end{theorem}

A first step to prove \autoref{thm:hypercubes} is to show that for a strategy profile of the form as described in \autoref{thm:hypercubes}, it holds that whenever a single player chooses another position, then for the resulting strategy profile the payoff of each player equals exactly the number of vertices to which she has a unique closest distance.
Note that we already know from \autoref{obs:uniqueshortestpath} that each player always colors all the vertices to which she has the uniquely closest distance. In \autoref{lem:standoff} we will show the opposite direction, namely that no player obtains vertices to which the distance is not the unique closest distance among all players (note that this does not hold in general).
In a second step, we can then compute the payoffs of all players and show
that they are maximal for the strategy profile stated in \autoref{thm:hypercubes}.

We start with the following lemma.

\begin{lemma}
  \label{lem:standoff}
  Let~$x,y\in\{0,1\}^d$ be two adjacent vertices in~$H_d$, $d\ge 1$, and let
  $(p_1,\ldots,p_4)$ be a strategy profile with $\{x,y,\overline{y}\}\subseteq\{p_1,\ldots,p_4\}$ and~$p_i\neq p_j$ for all~$i,j\in[4]$, $i\neq j$.
  Let~$v\in\{0,1\}^d$ be a vertex and let~$\delta:=\min_{i=1,\ldots,4}\ham(p_i,v)$.
  If there exist~$i,j\in[4]$, $i\neq j$, such that~$\ham(p_i,v) = \ham(p_j,v)=\delta$, then
  the vertex~$v$ will not be colored by any player at the end of the propagation process.
\end{lemma}

\begin{proof}
  Since the order of the players does not matter, we assume that~$p_1=p\in\{0,1\}^d$, $p_2=x$, $p_3=y$, and~$p_4=\overline{y}$. For the case~$\delta\le 1$, the statement clearly holds by definition of the propagation process. Hence, we consider~$\delta\ge 2$. Note that for any two vertices $p_i\neq p_j$ with~$\ham(p_i,v)=\ham(p_j,v)=\delta$, the distance~$\ham(p_i,p_j)$ must be even (and thus, at least two) since two vertices with an odd distance cannot have the same distance to any other vertex.
  Hence, $x$ and~$y$ can never have the same distance~$\delta$ to~$v$ since their distance is one (they are adjacent).
  It follows that either two or three players have the same shortest distance~$\delta$ to~$v$ among all players.
  For both cases, we show that there always exist two neighbors of~$v$ each having a different player with a unique closest distance of~$\delta-1$. Using \autoref{obs:uniqueshortestpath}, we can then conclude that these neighbors are colored in different colors and that~$v$ is thus removed.

  First, assume that exactly the two players~$i$ and~$j$ have distance~$\delta$ to~$v$, while the other two players have a distance larger than~$\delta$. Then,~$p_i$ and~$p_j$ differ in~$\ham(p_i,p_j)= 2c$ bits for some~$c\ge 1$. Note that~$v$ equals~$p_i$ in exactly~$c$ of these bits and equals~$p_j$ in the other~$c$ bits (otherwise they cannot have the same distance to~$v$).
  Hence, by swapping one of the~$c$ bits where~$v$ equals~$p_j$, we reach a neighbor~$u$ of~$v$ such that $\ham(p_i,u) = \delta-1 < \ham(p_j,u)=\delta+1$.
  Analogously, swapping one of the~$c$ bits where~$v$ equals~$p_i$ yields a neighbour~$w$ with $\ham(p_j,w) = \delta-1 < \ham(p_i,w)=\delta+1$.
  Note that the other two players have distance at least~$\delta$ to both~$u$ and~$w$ since their distance to~$v$ is at least~$\delta+1$.
  Thus, $p_i$ has the unique shortest distance to~$u$, and~$p_j$ has the unique shortest distance to~$w$. According to \autoref{obs:uniqueshortestpath},~$u$ and~$w$ are thus colored in different colors at time~$\delta-1$. Consequently, $v$ is removed at time~$\delta$.

  Now, assume that exactly three players have the minimum distance~$\delta$ to~$v$. Since~$x$ and~$y$ are adjacent, we know that the only two possible cases are that~$p$, $\overline{y}$, and either~$x$ or~$y$ have distance~$\delta$ to~$v$.

  \begin{enumerate}[\bf {Case} 1:]
  
  \item $\ham(p,v)=\ham(x,v)=\ham(\overline{y},v)=\delta < \ham(y,v)$.
  Note that~$\ham(y,v) = d - \ham(\overline{y},v)$, which yields~$\delta < d/2$.
  Note further that~$x$ and~$\overline{y}$ differ in an even number of~$\ham(x,\overline{y})=d-1$ bits, which means~$d$ is odd. Hence,~$\delta \le (d-1)/2$. Moreover,~$v$ equals each of~$x$ and~$\overline{y}$ in exactly half of their~$d-1$ differing bits (otherwise~$v$ cannot have the same distance to both). It follows that~$\delta = (d-1)/2$ and it also holds that~$x_i=\overline{y}_i\neq v_i$ is not possible for any~$i\in[d]$
  (this holds since $x_i = \overline{y}$ implied $\overline{y} = v_i$, which holds since $\delta = (d - 1) / 2$) .
  Now consider the vertex~$p$ which also differs in~$\delta$ bits to~$v$. Clearly,
  $p$ and~$v$ can neither differ in the same~$\delta$ bits as~$v$ and~$x$, nor in the same~$\delta$ bits as~$v$ and~$\overline{y}$ since then~$p_1$ would be equal to either~$x$ or~$\overline{y}$. Hence, there exist indices~$i$ and~$j$ among the~$d-1$ differing bits of~$x$ and~$\overline{y}$ such that~$v_i=\overline{y}_i=p_{i}\neq x_i$ and $v_j=x_j=p_{j}\neq \overline{y}_j$.
  Thus, by swapping the~$i$th bit of~$v$, we reach a neighbor~$u$ with~$\ham(x,u)=\delta-1 < \ham(p,u)=\ham(\overline{y},u)=\delta+1$. Similarly, swapping the~$j$th bit of~$v$ yields a neighbor~$w$ with~$\ham(\overline{y},w)=\delta-1 < \ham(p,w)=\ham(x,w)=\delta+1$. Clearly, also~$y$ has distance at least~$\delta$ to both~$u$ and~$w$. Hence, by \autoref{obs:uniqueshortestpath}, $u$ and~$w$ are colored in different colors and thus~$v$ is removed.

\item $\ham(p,v)=\ham(y,v)=\ham(\overline{y},v)=\delta < \ham(x,v)$.
  Note that~$\ham(y,\overline{y})=d$ is even and that~$v$ equals both~$y$ and~$\overline{y}$ in exactly~$\delta = d/2$ bits.
  As in Case~1, the vertex~$p$ cannot differ from~$v$ in the same~$\delta$ bits as~$y$ or~$\overline{y}$. Thus, we again find indices~$i$ and~$j$ with
  $v_i=\overline{y}_i=p_{i}\neq y_i$ and $v_j=y_j=p_{j}\neq \overline{y}_j$ such that for the corresponding neighbors~$u$ and~$w$ of~$v$ we have $\ham(y,u)=\delta-1 < \ham(p,u)=\ham(\overline{y},u)=\delta+1$ and $\ham(\overline{y},w)=\delta-1 < \ham(p,w)=\ham(y,w)=\delta+1$. Since by assumption also~$x$ has distance at least~$\delta$ from~$u$ and~$w$, it follows that~$u$ and~$w$ are colored in different colors and~$v$ is removed.
  \end{enumerate}
\end{proof}

\noindent\autoref{lem:standoff} (together with \autoref{obs:uniqueshortestpath}) shows that in every strategy profile as described in \autoref{thm:hypercubes}
the payoff of each player equals the number of vertices to which she has the unique minimum distance among all players.

The following lemma gives upper bounds on the possible payoffs for players in such a profile.

\newcommand{\nicev}{V_x^{y\overline{y}}}

\begin{lemma}
  \label{lem:closestvertexbound}
  Let~$x,y\in\{0,1\}^d$, $d\ge 1$, and let $\nicev:=\{v\in\{0,1\}^d \mid \ham(v,x)<\min\{\ham(v,y),\ham(v,\overline{y})\}$. Then, $|\nicev|\le 2^{d-2}$.
  Moreoever, if~$d$ is odd, then the bound is even smaller, that is,
  \[|\nicev| \le 2^{d-2} -\frac{1}{2}\binom{d-1}{(d-1)/2}.\]
\end{lemma}

\begin{proof}
  Let $\alpha:=\ham(x,y)$; thus,~$\ham(x,\overline{y})=d-\alpha$. For any vertex $v\in \nicev$, it holds $\ham(x,v) < \ham(y,v)$ and also $\ham(x,v) < \ham(\overline{y},v)$. The first inequality implies that~$v$ equals~$x$ in more than half of the~$\alpha$ bits where~$x$ and~$y$ differ. Analogously, the second inequality implies that $v$ equals~$x$ in more than half of the~$d-\alpha$ bits where~$x$ and~$\overline{y}$ differ. Clearly, the set of bits where~$x$ and~$y$ differ is disjoint from the set of bits where~$x$ and~$\overline{y}$ differ.
  Hence, the number of possible vertices in~$\nicev$ is
  \begin{align}
    |\nicev|&=\left(\sum_{\ell=\lceil(\alpha+1)/2\rceil}^\alpha\binom{\alpha}{\ell}\right) \cdot \left(\sum_{\ell=\lceil(d-\alpha+1)/2\rceil}^{d-\alpha}\binom{d-\alpha}{\ell}\right)\label{eq:bound}\\
    &\le 2^{\alpha-1}\cdot 2^{d-\alpha-1}=2^{d-2},\nonumber
  \end{align}
  which proves the general bound. Now, for~$d$ being odd, note that either~$\alpha$ or~$d-\alpha$ is even. We assume without loss of generality that~$\alpha\le d-1$ is even (that is, $d-\alpha$ is odd). Then, Equation~(\ref{eq:bound}) can be written as
  \begin{align*}
    |\nicev|&=\left(\sum_{\ell=\alpha/2+1}^\alpha\binom{\alpha}{\ell}\right) \cdot \left(\sum_{\ell=\lceil(d-\alpha)/2\rceil}^{d-\alpha}\binom{d-\alpha}{\ell}\right)\\
    &\le \left(2^{\alpha-1} - \frac{1}{2}\binom{\alpha}{\alpha/2}\right) \cdot 2^{d-\alpha-1}\\
    &= 2^{d-2} - \left(2^{d-\alpha-2}\cdot\binom{\alpha}{\alpha/2}\right).
  \end{align*}
  We now use the following identity for the central binomial coefficient
  \[\binom{\alpha}{\alpha/2}= \frac{\alpha!}{((\alpha/2)!)^2} = 2^\alpha \cdot \frac{1\cdot 3\cdot \ldots \cdot (\alpha-1)}{2 \cdot 4 \cdot \ldots \cdot \alpha}\]
  to obtain
  \begin{align*}
    |\nicev| &\le 2^{d-2} - \left(2^{d-2}\cdot\frac{1\cdot 3\cdot \ldots \cdot (\alpha-1)}{2 \cdot 4 \cdot \ldots \cdot \alpha}\right)\\
    &\le 2^{d-2} - \left(2^{d-2}\cdot\frac{1\cdot 3\cdot \ldots \cdot (d-2)}{2 \cdot 4 \cdot \ldots \cdot (d-1)}\right)\\
    &\le 2^{d-2} - \frac{1}{2}\binom{d-1}{(d-1)/2}.
  \end{align*}
\end{proof}

We can now proceed with proving \autoref{thm:hypercubes}.

\begin{proof}[Proof of \autoref{thm:hypercubes}]
  To start with, observe that for~$d=1$ the statement clearly holds since there are only two vertices in~$H_1$, which  gives a payoff of zero for each player, and it is not possible to obtain more than zero vertices for any player (by definition of the diffusion game).

  In the following we consider~$d\ge 2$.
  Since the ordering of the players does not matter, we fix the strategy profile~$(p_1=x,p_2=\overline{x},p_3=y,p_4=\overline{y})$. Moreover,
  due to the symmetry of the hypercube, we only have to consider the case that Player~1 changes her strategy. Let us first determine the payoff of Player~1 for the above profile.
  According to \autoref{obs:uniqueshortestpath} and \autoref{lem:standoff}, we know that Player~1 obtains exactly those vertices~$v\in\{0,1\}^d$ to which she has the unique minimum distance, that is, $\ham(x,v)<\min\{\ham(\overline{x},v),\ham(y,v),\ham(\overline{y},v)\}$.
  Recall that~$x$ and~$y$ are adjacent, that is, they differ in~$\ham(x,y)=1$ bit.
  Therefore, $\ham(x,v)<\ham(y,v)$ implies that~$v$ equals~$x$ in that bit.
  Thus, we have~$\ham(y,v)=\ham(x,v)+1$, $\ham(\overline{x},v)=d-\ham(x,v)$, and~$\ham(\overline{y},v)=d-1-\ham(x,v)$. Hence, $v$ has to satisfy~$\ham(x,v)<(d-1)/2$ in order to satisfy $\ham(x,v)<\ham(\overline{y},v)$ and $\ham(x,v)<\ham(\overline{x},v)$.
  That is, $v$ is allowed to differ from~$x$ in at most~$\lfloor(d-1)/2\rfloor$ of the~$d-1$ bits where~$x$ is equal to~$y$. The payoff of Player~1 is thus \[\sum_{\ell=0}^{\lfloor(d-1)/2\rfloor}\binom{d-1}{\ell}.\]
  Note that this payoff equals~$2^{d-2}$ if~$d$ is even, and $2^{d-2} - \frac{1}{2}\binom{d-1}{(d-1)/2}$ if~$d$ is odd.

  Now, let Player~1 choose an arbitrary vertex~$p_1\in\{0,1\}^d$. Clearly, we can assume that $p_1\not\in\{\overline{x},y,\overline{y}\}$ since her payoff is zero otherwise.
  Hence, by \autoref{obs:uniqueshortestpath} and \autoref{lem:standoff}, we know again that the payoff of Player~1 equals the number of vertices in $\{v\in\{0,1\}^d \mid \ham(x,v)<\min\{\ham(\overline{x},v),\ham(y,v),\ham(\overline{y},v)\}\}$.
  By \autoref{lem:closestvertexbound}, we know that this number is at most~$2^{d-2}$ if $d$ is even, and at most $2^{d-2} - \frac{1}{2}\binom{d-1}{(d-1)/2}$ if $d$ is odd.
  Therefore, Player~1 cannot increase her payoff by changing her strategy, which finishes the proof.
\end{proof}

The existence of \nashs on hypercubes for three players as well as for more than four players  remains open in general.

\section{General Graphs}
\label{sec:general}

  \begin{figure}[h]
    \center
    \begin{tikzpicture}[draw=black!75, scale=0.75,-]
      \tikzstyle{vertex}=[circle,draw=black!80,minimum size=12pt,inner sep=0pt]

      \foreach [count=\i] \pos / \text in {
        {(0,2)}/,
        {(1,1)}/,
        {(1,3)}/,
        {(2,1.5)}/,
        {(2,2.5)}/,
        {(3,1)}/,
        {(3,3)}/,
        {(4,2)}/       
      }
      {
        \node[vertex] (V\i) at \pos {\text};
      }

      \foreach \i / \j in {1/2,2/4,4/5,5/3,3/1,4/6,6/8,8/7,7/5,2/6,3/7,6/5,7/4} {
        \path[] (V\i) edge (V\j);
      }
    \end{tikzpicture}
    \caption{A graph on 8 vertices with no \nash for two players.}\label{fig:weirdgraph}
  \end{figure}
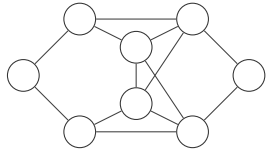

In this section, we study the existence of \nashs
on arbitrary graphs.
Using computer simulations, we found that for two players, a \nash
exists on any graph with at most~$n=7$ vertices.
For~$n=8$, we obtained the graph depicted in~\autoref{fig:weirdgraph},
for which there is no \nash for two players.
As it is clear that adding isolated vertices
to the graph in~\autoref{fig:weirdgraph} does not allow
for a \nash, we conclude the following.

\begin{corollary}
  For two players, there is a \nash on each $n$-vertex graph if
  and only if~$n \le 7$.
\end{corollary}

For more than two players, we can show the following.

\begin{theorem}\label{thm:general}
  For any $k > 2$ and any $n\ge\floor{\frac{3}{2}k} + 2$,
  there exists a tree with $n$ vertices such that there
  is no \nash for $k$ players.
\end{theorem}

\begin{proof}
  We describe a construction only for $n=\floor{\frac{3}{2}k} + 2$,
  as we can add arbitrarily many isolated vertices without introducing
  a \nash.

  We first describe the construction for $k$ being odd.
  We create one $P_3$,
  whose vertices we denote by $u_1$, $u_2$, and $u_3$,
  such that $u_2$ is the middle vertex of this~$P_3$.
  For each $i \in [2, \ceil{\frac{k}{2}}]$,
  we create a copy of $P_3$,
  denoted by $P_i$,
  whose vertices we denote by $v_{i, 1}$, $v_{i, 2}$, and $v_{i, 3}$,
  such that $v_{i, 2}$ is the middle vertex of $P_i$.
  For each $i \in [2, \ceil{\frac{k}{2}}]$,
  we connect $v_{i, 1}$ to $u_3$.
  An example for~$k = 9$ is depicted in \autoref{fig:seven}.

  \begin{figure}[t]
    \center
    \begin{tikzpicture}[draw=black!75, scale=0.75,-]
      \tikzstyle{vertex}=[circle,draw=black!80,minimum size=12pt,inner sep=0pt]

      \foreach [count=\i] \pos / \text in {
        {(0,1.5)}/$u_1$,
        {(1,1.5)}/$u_2$,
        {(2,1.5)}/$u_3$,
        {(3,0)}/$v_{1,1}$,
        {(4,0)}/$v_{1,2}$,
        {(5,0)}/$v_{1,3}$,
        {(3,1)}/$v_{2,1}$,
        {(4,1)}/$v_{2,2}$,
        {(5,1)}/$v_{2,3}$,
        {(3,2)}/$v_{3,1}$,
        {(4,2)}/$v_{3,2}$,
        {(5,2)}/$v_{3,3}$,
        {(3,3)}/$v_{4,1}$,
        {(4,3)}/$v_{4,2}$,
        {(5,3)}/$v_{4,3}$}
      {
        \node[vertex] (V\i) at \pos {\text};
      }

      \foreach \i / \j in {1/2,2/3,3/4,4/5,5/6,3/7,7/8,8/9,3/10,10/11,11/12,3/13,13/14,14/15} {
        \path[] (V\i) edge (V\j);
      }
    \end{tikzpicture}
    \caption{A tree with no \nash for $9$ players.}\label{fig:seven}
  \end{figure}
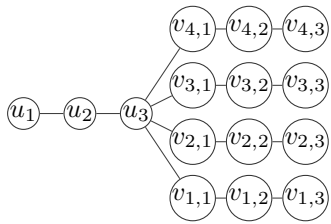

  To see that there is no \nash for the constructed graph,
  consider first strategy profiles for which $u_3$ is free
  (that is, no player chooses $u_3$).
  If also both $u_1$ and $u_2$ are free,
  then there exists some $P_i$ with at least $2$ occupied vertices
  (by the pigeon-hole principle).
  It is clear that at least one of the players occupying these vertices
  can increase her payoff by moving to $u_3$.
  If $u_1$ and $u_2$ are both occupied,
  then there exists some $P_i$ with at most $1$ occupied vertex
  (again, by the pigeon-hole principle).
  It is clear that the player occupying this vertex
  can increase her payoff by moving to $u_3$.
  If only one vertex out of $u_1$ and~$u_2$ is occupied,
  then the player occupying this vertex can
  increase her payoff by moving to $u_3$.
  Therefore,
  we can assume that $u_3$ is occupied.
  In this case,
  it holds that if both $u_1$ and $u_2$ are free,
  then at least one player has a payoff of $1$,
  and she gains more by moving to $u_2$.
  If both $u_1$ and $u_2$ are occupied,
  then at least one $P_i$ has at most one occupied vertex,
  and this occupied vertex can only be $v_{i, 1}$,
  therefore the player occupying $u_2$ gains more by moving to~$v_{i, 2}$.
  Lastly,
  if exactly one out of $u_1$ and $u_2$ is free,
  then at least one $P_{i_1}$ has at most one occupied vertex,
  and this occupied vertex can only be $v_{i_1, 1}$.
  Moreover,
  at least one $P_{i_2}$ has at least two occupied vertices,
  therefore a player occupying one of these vertices
  gains more by moving to $v_{i_1, 2}$.
  Therefore,
  this graph has no \nash for $k$ players.

  For~$k$ being even,
  we create one $P_2$,
  whose vertices we denote by $u_1$, $u_2$.
  For each $i \in [2, \frac{k}{2}+1]$,
  we create a copy of $P_3$,
  denoted by $P_i$,
  whose vertices we denote by $v_{i, 1}$, $v_{i, 2}$, and $v_{i, 3}$,
  such that $v_{i, 2}$ is the middle vertex of $P_i$.
  For each $i \in [2, \frac{k}{2}+1]$,
  we connect $v_{i, 1}$ to $u_2$.
  This graph has no \nash for $k$ players,
  as can be verified by a similar analysis as above.
\end{proof}

\section{Conclusion}

We studied a competitive diffusion game for three or more players
on several classes of graphs, answering---as a main contribution---an
open question concerning the existence of a \nash for three players on grids \cite{roshanbin14} negatively.
Further, extending previous results on hypercubes~\cite{etesami2014complexity},
we proved that \nashs always exist for four player on $d$-dimensional hypercubes.
With this work, we provide a first systematic study of this game for more than two players.
However, there are several questions left open, of which we mention
some here.

An immediate question (generalizing~\autoref{thm:grids})
is whether a \nash exists for more than three players on a grid.
Computer simulations lead us to the conjecture that there is no \nash for four players on a grid of size larger than~$6\times 6$.
A further immediate question (generalizing~\autoref{thm:hypercubes})
is whether a \nash exists for three players or more than four players on a $d$-dimensional hypercube.
Also, giving a lower bound for the number of vertices~$n$ such that
there is a graph with~$n$ vertices with no \nash for~$k$ players
is an interesting question as it is not clear that the upper bounds
given in~\autoref{thm:general} are optimal.
In other words, is it true that $n \leq \frac{3}{2}k+1$ implies the
existence of a \nash for~$k$ players?

\paragraph*{Acknowledgement.}
We thank Manuela Hopp for helpful discussions and implementations for the result on hypercubes.
Laurent Bulteau was supported by the Alexander von Humboldt Foundation, Bonn, Germany.
Vincent Froese was supported by the DFG, project DAMM (NI 369/13).
Nimrod Talmon was supported by DFG Research Training Group ``Methods for Discrete Structures''~(GRK~1408).

\makeatletter
\renewcommand\bibsection%
{
  \section*{\refname
    \@mkboth{\MakeUppercase{\refname}}{\MakeUppercase{\refname}}}
}
\makeatother

\bibliographystyle{abbrvnat}
\bibliography{bibliography}

\end{document}